\documentclass[superscriptaddress,aps,pra,amssymb,amsfonts, twocolumn]{revtex4}
\usepackage{amsmath} \usepackage{amsthm} \usepackage{color}
\usepackage{graphics,graphicx}
\usepackage[pdftex]{hyperref}

\usepackage{times}
\usepackage{amssymb,latexsym,graphicx}

\newcommand{\ket}[1]{|#1\rangle}
\newcommand{\bra}[1]{\langle#1|}

\newcommand{\proj}[1]{|#1\rangle\langle#1|}

\newcommand{\Tr}{\mbox{\rm Tr}}

\newcommand{\beq}{\begin{equation}}
\newcommand{\eeq}{\end{equation}}
\newcommand{\beqn}{\begin{equation*}}
\newcommand{\eeqn}{\end{equation*}}

\newcommand{\C}{\ensuremath{\mathbb{C}}}

\newtheorem{theorem}{Theorem}

\newtheorem{lemma}[theorem]{Lemma}
\newtheorem{claim}[theorem]{Claim}

\newtheorem{definition}[theorem]{Definition}

\newcommand{\be}{\begin{eqnarray}}
\newcommand{\ee}{\end{eqnarray}}

\newcommand{\Id}{\mathsf{id}}

\newcommand{\ignore}[1]{}

\newenvironment{feasSDP}[2]{
\smallskip
\begin{center}
\begin{tabular}{ll}
#1 & #2\\
whenever
}
{
\end{tabular}
\end{center}
\smallskip
}

\newcommand{\Real}{\mathbb{R}}
\newcommand{\hil}{\mathcal{H}}
\newcommand{\tr}{\mathsf{tr}}
\newcommand{\id}{\mathsf{id}}

\begin{document}

\title{More non-locality with less entanglement}
\author{Thomas Vidick}
\affiliation{Computer Science division, UC Berkeley, 94720 Berkeley, USA}
\email{vidick@eecs.berkeley.edu}
\author{Stephanie Wehner}
\affiliation{Center for Quantum Technologies, National University of Singapore, 2 Science Drive 3, 117543 Singapore}
\email{wehner@nus.edu.sg}
\date{\today}

\begin{abstract}
	We provide an explicit example of a Bell inequality with 3 settings and 2 outcomes per site for which the largest violation is not obtained by the maximally entangled state, even if its
	dimension is allowed to be arbitrarily large. 
	This complements recent results by Junge and Palazuelos (arXiv:1007.3042) who show, employing tools from operator space theory, that such inequalities do 
	\emph{exist}.
	Our elementary example provides a simple, natural setting in which it can be explicitly demonstrated that even an arbitrarily large supply of EPR pairs is not the strongest possible nonlocal resource. 
\end{abstract}
\maketitle

\section{Introduction}

Entanglement is a powerful resource, facilitating computation, communication, or more generally any nonlocal task. 
Like all resources it is useful to be able to measure it, so that entangled states could be ranked according to their usefulness for a given task. 
A very natural measure for the entanglement of any bipartite state $\ket{\Psi} \in \hil_A \otimes \hil_B$ is the
entropy of entanglement
$E(\Psi) = S(\rho_A)$~\cite{bennett:entEntropy}, where $S(\rho_A) = - \Tr(\rho_A \log \rho_A)$ is the von Neumann entropy
and $\rho_A = \tr_B(\proj{\Psi})$ is the reduced density operator of $\ket{\Psi}$ on one of the two subsystems. 
In any dimension $d$ this measure is maximized by the maximally entangled state 
\begin{align}
	\ket{\Psi_d} = \frac{1}{\sqrt{d}}\sum_{j=1}^d \ket{j}\ket{j}\ .
\end{align}
Since $\ket{\Psi_d}$ exhibits the largest amount of entanglement, it would be natural to guess that it would indeed
be the most useful state for any nonlocal task. This belief is reinforced by the fact that this state has proven extremely 
useful for many quantum information problems (e.g.~\cite{teleport,e91,superdense}), and is by itself a sufficient resource for 
the creation of \emph{any} other nonlocal state as soon as one allows local operations and classical communication (LOCC)~\cite{teleport}. 
Moreover, it is known that any shared pure state $\ket{\Psi}$ violates a Bell inequality if and only if it is entangled~\cite{gisin:nonProduct,popescu:generic}, 
suggesting that the amount of entanglement may play a central role in quantifying the strength of nonlocal correlations.

For a long time it was implicitly assumed that $\ket{\Psi_d}$ is the most useful state with respect to violation of Bell inequalities~\cite{bell}.
The first doubts cast on this conjecture stem from a result by Eberhard~\cite{eberhard:background} who showed that when it comes to closing the detection efficiency loophole 
less entangled states can be more useful. More recently, such doubts were confounded by the surprising fact that, at least in small dimensions in which numerical experiments can be 
conducted, 
there are inequalities for which the maximally entangled state does not give the maximum violation. More specifically, for 
every dimension $d$ there is a Bell inequality (such as the CGLMP inequality~\cite{Collins2002}) which in that dimension is maximally violated by a state different 
from the maximally entangled state
--- a state with \emph{lower} entanglement~\cite{acin:qutrit,acin:maxEntangled,richard:cglmp}.
Conversely, unlike for the maximally entangled state, 
even an infinite supply of certain maximally non-local resources does not allow us to simulate all possible correlations coming from some less entangled states~\cite{nicolas:differentResources}.
This has prompted the realization that \emph{nonlocality} might be a resource of a different nature than entanglement, and many other examples have been discovered in the realm of
Bell inequalities and quantum cryptography (see~\cite{Methot2006} for a survey), as well as quantum information theory~\cite{Bennett2009,mario:reverseShannon}.
Since Bell inequalities provide a very natural setting in which the properties of entanglement manifest themselves, it is interesting to 
ask whether an arbitrarily large supply of EPR pairs could indeed not be a sufficient resource for their maximal violation. 
In a recent paper, Junge and Palazuelos~\cite{marius:maxEntangled} provided a negative answer to this question. In particular, they 
showed using tools employed in the study of operator algebras 
and a probabilistic argument that there \emph{exists} a family of Bell inequalities for which the maximum violation cannot be obtained by using the maximally entangled state, even if its dimension $d$
is arbitrarily large. However, they do not provide any explicit such inequality, and moreover their results only hold in an asymptotic sense, so that it is not immediately clear what is the size of the smallest Bell inequality for which this phenomenon can be observed.

\subsection{Result}

In this note, we provide a simple example of an extremal Bell inequality that by itself already 
demonstrates that it is sometimes necessary to use \emph{less} entanglement in order to obtain \emph{more} nonlocality. 
Unlike the results of~\cite{marius:maxEntangled} we use only elementary techniques.

The inequality we consider involves measurements on only two sites, Alice and Bob,
where each party has only three measurement settings and two outcomes. This inequality is thus one of the simplest possible examples of a 
Bell inequality where an arbitrary amount of EPR pairs could really make a difference. First of all, it is known that any inequality with only two settings
and two outcomes per site can be maximally violated when Alice and Bob each hold only a single qubit~\cite{masanes05:_extrem}. In particular, this
means that if one were merely interested in knowning whether the maximally entangled state is 
\emph{necessary} to obtain the maximum possible violation it suffices to check states in a fixed dimension, where it is already known 
that the maximally entangled state is not always the best. Second, however, the inequality we consider is also the first \emph{extremal} Bell inequality
after the CHSH inequality~\cite{chsh}. 
The CHSH inequality is particularly fundamental in that 
violation (by a certain state) of any two-setting inequality implies the same state also violates CHSH. 
Moreover, it is known that the CHSH inequality is maximally violated by a single EPR pair. When one allows three settings, the only new independent inequality (i.e., which can be violated by a state not violating the CHSH inequality) is the one we consider~\cite{CG}.

More specifically, the inequality we consider is the well-known $I_{3322}$ inequality, first introduced in~\cite{Froissart1981}. We will use 
$\{A_j\}_{j \in \{1,2,3\}}$ and $\{B_k\}_{k \in \{1,2,3\}}$ to denote the measurement operators for the first of the two possible outcomes for Alice and Bob respectively.
Using the common shorthands 
\begin{align}
	\langle A_j B_k\rangle &:= \bra{\Psi}A_j \otimes B_k \ket{\Psi}\ ,\\
	\langle A_j \rangle &:= \bra{\Psi}A_j \otimes \id \ket{\Psi}\ ,\\
	\langle B_k\rangle &:= \bra{\Psi} \id \otimes B_k \ket{\Psi}\ ,
\end{align}
we define
\begin{align}
\langle I_{3322} \rangle := &-\langle A_2 \rangle - \langle B_1 \rangle - 2 \langle B_2 \rangle + \langle A_1 B_1 \rangle \nonumber\\
&+ \langle A_1 B_2 \rangle
+ \langle A_2 B_1 \rangle + \langle A_2 B_2 \rangle - \langle A_1 B_3 \rangle \nonumber\\
&+ \langle A_2 B_3 \rangle
- \langle A_3 B_1 \rangle + \langle A_3 B_2 \rangle
	\label{eq:I3322}
\end{align}
While for classical correlations we have 
\begin{align}\label{eq:classicalBound}
	\langle I_{3322} \rangle \leq 0\ ,
\end{align}
there exist measurements~\cite{CG} such that using just one EPR pair (i.e. $\ket{\Psi} = \ket{\Psi_2}$) one can get
\begin{align}\label{eq:oneEPR}
	\langle I_{3322} \rangle = \frac{1}{4}\ .
\end{align}
Yet, the precise maximum of $\langle I_{3322} \rangle$ over all quantum states and measurements remains unknown. Numerical upper-bounds were obtained using a SDP hierarchy in~\cite{as:qmp}. This was followed by recent exhaustive numerical investigations by P\'al and V\'ertesi~\cite{Pal2010}, who report very interesting results. 
Their experiments suggest that the optimum violation of~\eqref{eq:classicalBound}, even though it only involves a constant number of settings and outcomes, might only be reached in infinite dimension. Indeed, they find strategies obtaining a value of at least $0.25084...$ (matching the upper bound up to precision $10^{-7}$ in dimension $\approx 100$), and moreover in their experiments this value keeps increasing as the dimension of the strategies is allowed to increase.
Moreover, even though the observables which achieve the maximum violation in a given dimension have a rather simple and systematic form, 
the corresponding state has an interesting distribution of Schmidt coefficients, and it is quite far from the maximally entangled state. 
Hence their results suggest a simple inequality for which the maximally entangled state of \emph{any dimension} may not permit the largest violation, and moreover that the
maximal violation cannot be attained with any finite-dimensional state.

Here, we prove that indeed the maximally entangled state does not lead to the optimal violation of even such a simple inequality. 
That is,
\begin{theorem}\label{lem:main} For all dimensions $d \geq 0$, and any observables, using the maximally entangled state $\ket{\Psi} = \ket{\Psi_d}$ can lead to a violation
	of at most
	\begin{align}
		\langle I_{3322} \rangle \leq \frac{1}{4}\ .
	\end{align}
\end{theorem}
Note that in contrast with previous work,~\eqref{eq:oneEPR} tells us that a value of $1/4$ can be attained using just one EPR pair, and hence the maximally entangled state in any dimension is no more powerful
than the maximally entangled state for $d=2$.
This definitively demonstrates, in the simplest possible setting, that maximally entangled states are not the most nonlocal.

\subsection{Generic states}

Before embarking on our proof, it is worth pointing out that there does in fact exist a generic family of states that always allow us to obtain the maximum violation for \emph{any} Bell inequality.
These states, however, exhibit less entanglement than the maximally entangled state of same dimension. This ``universal'' family of states are known as \emph{embezzlement states}~\cite{patrick:embezzle}. They previously played an important role
in more involved tasks in quantum information theory, namely the so-called quantum reverse Shannon theorem~\cite{Bennett2009,mario:reverseShannon}, 
which provided another example where the maximally entangled state is not
sufficient to achieve the corresponding channel simulation result, but the universal embezzlement states are. 
The key property of the $d$-dimensional embezzlement state $\ket{\Phi_d}$ that is used is that, for \emph{any} pure state $\ket{\Psi}$, there exists $d$ and $d'$ such that $\ket{\Phi_d} \approx \ket{\Phi_{d'}}\otimes \ket{\Psi}$, where the equivalence only requires the application of local unitaries on each system; no communication is 
needed~\cite{patrick:embezzle}. Since an embezzlement state can be used to obtain any other pure state by local unitary operations, it immediately follows
that any Bell inequality can be maximally violated by an embezzlement state (of possibly higher dimension), as pointed out recently in~\cite{Oliveira2010}.
This demonstrates that, even though in small dimensions it might seem like every inequality has its own specialized maximizing state, if one allows the dimension to grow larger, then a simple class of states is sufficient to obtain maximal violations.  

\section{Using the maximally entangled state}

We now give a detailed overview of the proof of our main result (Theorem~\ref{lem:main}), relegating technical details to the appendix. 
Throughout we will refer to a particular choice of measurements applied to the maximally entangled state as a \emph{strategy}. Since our game is binary, it is known that  
we may assume without loss of generality (and without affecting the underlying state) that the operators used by Alice and Bob are projectors~\cite[Proposition 2]{Cleve2004}, and we will denote them by
$\{A_j,\id - A_j\}$ for Alice and $\{B_k,\id - B_k\}$ for Bob.
We will also refer to
\begin{align}
	\omega := \langle I_{3322} \rangle 
\end{align}
as the \emph{value} of a particular strategy.
Our goal is to show that $\omega$ is at most $1/4$, irrespective of the dimension $d$. 
We first introduce an important tool in our analysis, the CS decomposition of a pair of projectors. This decomposition was also at the heart of the results in~\cite{masanes05:_extrem}, where it was used to handle the case of only \emph{two} observables per site.

{\bf The CS decomposition.} Given a pair of $d$-dimensional projectors $P$ and $Q$, there exists an orthonormal basis in which the two projectors are jointly block-diagonal (see for instance~\cite{Bhatia:97a}). Moreover, the blocks can be either $1$-dimensional, in which case $P$ and $Q$ either have a $0$ or a $1$ in that block, or $2$-dimensional, in which case they can be written in the form
\begin{align}
P &= \frac{1}{2}\begin{pmatrix} 1-c & -s \\ -s & 1+c \end{pmatrix}\ ,\label{eq:cs1}\\
Q &= \frac{1}{2}\begin{pmatrix} 1-c & s \\ s & 1+c \end{pmatrix}\ ,\label{eq:cs2}
\end{align}
for some coefficients $c\in(-1,1)$ and $s=\sqrt{1-c^2}$. The angles $\theta$ such that $c=\cos \theta$ are called the \emph{principal angles} between the subspaces on which $P$ and $Q$ project. 

\medskip

Our proof proceeds in two steps. Step 1 is to show that we can greatly simplify the form of Alice's and Bob's measurement operators. 
The main idea is to show using the CS decomposition that for any strategy maximizing~\eqref{eq:I3322} there exists a basis in which all measurements are 
tridiagonal~\footnote{A matrix is tridiagonal if its only non-zero entries are on the main diagonal and the two diagonals right above and under it.}.
This lets us greatly reduce the number of parameters and give a relatively simple analytic expression for the value $\omega$ of the strategy.
Step 2 consists in upper-bounding this simple expression using standard analytic techniques.
 
\subsection{Step 1: A simple joint normal form}

This is arguably the most crucial step in our proof, as it lets us show that a completely arbitrary strategy 
given by projectors $\{A_j,B_k\}_{j,k=1,\ldots,3}$ can be put into a much simpler form without decreasing its value. 
As we mentioned previously, the key idea is to apply the CS decomposition twice, once to the pair $(A_1,A_2)$, and once to the pair $(B_1,B_2)$. This results in two orthonormal bases $\mathcal{B}_A$ and $\mathcal{B}_B$ such that the matrices of $(A_1,A_2)$ in $\mathcal{B}_A$ are block-diagonal, with blocks of the form~\eqref{eq:cs1} for $A_1$ and~\eqref{eq:cs2} for $A_2$, and similarly for $(B_1,B_2)$ in $\mathcal{B}_B$. We number the blocks of $(A_1,A_2)$ using even indices $2,\ldots,d$ and call the corresponding coefficients $c_{2i},s_{2i}$; the blocks of $(B_1,B_2)$ are numbered using odd indices $1,\ldots,d+1$ and corresponding coefficients $c_{2i+1},s_{2i+1}$.

In general the bases $\mathcal{B}_A$ and $\mathcal{B}_B$ are unrelated, but we argue that, under the condition that the strategy maximizes~\eqref{eq:I3322}, they must in fact be permutations of one another. 
To see this, note that \eqref{eq:I3322} can be re-written as
\begin{align}
&\langle I_{3322} \rangle =\nonumber \\
& \langle A_1+A_2 , B_1+B_2 \rangle + \langle A_2-A_1,B_3\rangle + \langle A_3,B_2-B_1\rangle\nonumber\\
 &-\, \langle A_2,\id\rangle - \langle \id,B_1\rangle -2\langle \id,B_2\rangle \label{eq:secondPart}
\end{align}
where 
\begin{align}
	\langle A,B\rangle = \frac{1}{d}\Tr(A^TB)
\end{align}
and we used that if $\ket{\Psi}$ is the maximally entangled state then
\begin{align}
	\bra{\Psi} A\otimes B \ket{\Psi} = \langle A,B\rangle\ .
\end{align}
Note that since the $A_j$ operators always appear on the left of the tensor product (Alice's side), we will henceforth argue about $A_j^T$ rather than $A_j$, omitting the transpose sign for simplicity of notation. 
For the moment, let's ignore the contribution of the last three terms in~\eqref{eq:secondPart}.  Observe that $A_3$ (resp. $B_3$) only appears in the term $\langle A_3,B_2-B_1\rangle$ (resp. $\langle A_2-A_1,B_3\rangle$). When maximizing over $A_3$ it is thus clear that the optimal choice is to make 
$A_3$ the projector onto the positive eigenspace of 
$B_2-B_1$ (resp. $B_3$ to project on the positive eigenspace of $A_2-A_1$). 
This in particular implies that the value of those two terms is \emph{independent} of the choice of $\mathcal{B}_B$ (resp. $\mathcal{B}_A$). 
Hence the choice of the bases $\mathcal{B}_A$, $\mathcal{B}_B$ only bears influence on the value of the first term in~\eqref{eq:secondPart}. 

Let us now examine the first term. Note that the precise form~\eqref{eq:cs1},~\eqref{eq:cs2} in which we wrote the CS decomposition ensures that $A_1+A_2$ is diagonal in $\mathcal{B}_A$ (resp. $B_1+B_2$ in $\mathcal{B}_B$). It is well known (see Claim~\ref{claim:basis} in the appendix) that $\langle A_1+A_2 , B_1+B_2 \rangle$ is maximized whenever the vectors in $\mathcal{B}_B$ are a permutation of those in $\mathcal{B}_A$.
It follows that for the optimal choice of bases $A_1+A_2$ and $B_1+B_2$ will necessarily be simultaneously diagonal. 

However, this does not necessarily imply that the blocks of $(A_1,A_2)$ are aligned with those of $(B_1,B_2)$, as corresponding pairs of basis vectors need not match ---  in fact, if they did, then it is not hard to see that the strategy would be reduced to a convex combination of $2$-dimensional strategies, which would conclude our proof. Nevertheless, by a simple argument we can show that without loss of generality the blocks are simply ``shifted'': there exists an ordering of $\mathcal{B}_A = \{e_1,\ldots,e_d\}$ such that if the blocks of $(A_1,A_2)$ correspond to pairs $(e_1,e_2),(e_3,e_4),\ldots$ then those of $(B_1,B_2)$ can be seen to correspond to pairs $(e_d,e_1),(e_2,e_3),\ldots$. 

The exact form we obtain for the strategies is given in Definition~\ref{def:normalform} in the Appendix, and gaps in the argument above are filled in the proof of Lemma~\ref{lem:blocks}, which can informally be summarized as follows.

\begin{lemma}[Lemma~\ref{lem:blocks}, informal] There exists a basis $(e_1,\ldots,e_d)$ in which
\begin{itemize}
\item $(A_1,A_2,B_3)$ (resp. $(B_1,B_2,A_3)$) are jointly block-diagonal.
\item The blocks corresponding to each of these decompositions are shifted: blocks of $(A_1,A_2,B_3)$ correspond to pairs $(e_{2i-1},e_{2i})$, while blocks of $(A_1,A_2,B_3)$ correspond to pairs $(e_{2i},e_{2i+1})$.
\item  The blocks of $(A_1,A_2)$ are of the form~\eqref{eq:cs1},~\eqref{eq:cs2} with coefficients $(c_{2i},s_{2i})$, $i=1,\ldots,d/2$, while those of $(B_1,B_2)$ are of the same form with corresponding coefficients $(c_{2i+1},s_{2i+1})$, $i=0,\ldots,d/2-1$. 
\end{itemize}
\end{lemma}

\subsection{Step 2: The value of a strategy in joint normal form}

 Once we have found a nice basis in which to express all observables appearing in the strategy, it should appear as no surprise that the value of~\eqref{eq:I3322} should be easily expressible as 
 a function of the coefficients $(c_i)_{i=1,\ldots,d}$, since these are the only free parameters left in our choice of strategy. In fact, fixing coefficients $c_i$ where $i$ is even, it 
 is not hard to determine the optimal choice of coefficients $c_i$ for odd $i$. This reduces the size of our problem to the $d/2$ parameters $c_2,\ldots,c_{d}$.
 One can then show that the strategy has the following value (cf. Lemma~\ref{lem:bound} for a more precise statement):
 \begin{align}\label{eq:omegaintro}
 \omega = \frac{1}{d}\sum_{i=1}^{d/2} f(c_{2i-1},c_{2i+1}) + \frac{c_1-c_{d+1}}{2d}\ ,
 \end{align}
where 
\begin{align*}
	f(x,y)
	= \sqrt{(x+y)^2+1}+ \frac{1}{2}\sqrt{1-x^2} + \frac{1}{2}\sqrt{1-y^2} - 2\ .  
\end{align*}
We have thus rephrased the problem of maximizing $\langle I_{3322} \rangle$ over all strategies in terms of maximizing  $\omega$ over 
all admissible coefficients $(c_{2i-1})_{i=1,\ldots,d/2+1}$. To prove our claim, it only remains to prove an upper bound on $\omega$, which can be done using standard 
analytical techniques provided in the appendix.

\begin{lemma}\label{lem:numerics} 
	Let $c_{2i-1} \in [-1,1]$, for $i=1,\ldots, d/2+1$. Then the expression $\omega = \omega(c_i)$ in~\eqref{eq:omegaintro} is upper-bounded by $\frac{1}{4}$.
\end{lemma}

\section{Conclusion and open questions}

We have provided a concrete example of a simple inequality for which it can be shown that the maximally entangled state of any dimension is not the most nonlocal state. An interesting question, already asked in~\cite{Pal2010}, is whether one can show that optimal violation of the $I_{3322}$ inequality requires a state of infinite dimension. This is strongly suggested by the strategies found numerically by Pal and Vertesi, which, even though they are based on an entangled state which is very far from the maximally entangled state, have a matrix form which is quite similar to the one in Def.~\ref{def:normalform}. 
Extending our argument to show that Alice and Bob's measurements always have this form, even when they do not use the maximally entangled state, 
would be a big step towards proving that no finite-dimensional strategy is optimal~\cite{salman:personal}.
This would not only have very interesting consequences for our understanding of Bell inequalities, but also for the optimization of polynomials with non-commutative variables.
In particular, it would imply that the SDP hierarchies suggested in~\cite{npa1,npa2,as:qmp} only converge in the limit of infinitely many levels, which is an open problem even outside 
the realm of quantum information.

\smallskip
\noindent
{\bf Note:} After posting our work to the arXiv, we learned about related work by Liang et al.~\cite{Liang2010}, which by now has also appeared
on the arXiv. In a different context and independently of our results, they found a two-setting, two-outcome 
linequality for which the maximally entangled state of any dimension is not the most nonlocal. Recall, however, that in that case it is already known that 
 local dimension $2$ is sufficient for both Alice and Bob, and hence it suffices to determine whether maximal entanglement in that fixed dimension 
is necessary to obtain the maximum violation. 
The motivation for their work, however, is in a different context in which such inequalities are indeed interesting. 

\acknowledgments

We thank Salman Beigi for interesting discussions, and Oded Regev for helpful comments. TV was supported by ARO Grant W911NF-09-1-0440 and NSF Grant CCF-0905626. 
SW was supported by the National Research Foundation (Singapore), and the Ministry of Education (Singapore). 
TV is grateful to CQT, Singapore, for hosting him while part of this work was done.

\bibliography{I3322}

\begin{thebibliography}{10}

\bibitem{acin:qutrit}
A.~Acin, T.~Durt, N.~Gisin, and J.~I. Latorre.
\newblock Quantum non-locality in two three-level systems.
\newblock {\em Phys. Rev. A}, 65:052325, 2002.

\bibitem{acin:maxEntangled}
A.~Acin, R.~Gill, and N.~Gisin.
\newblock Optimal bell tests do not require maximally entangled states.
\newblock {\em Phys. Rev. Lett.}, 95:210402, 2005.

\bibitem{salman:personal}
S.~Beigi.
\newblock Personal communication, 2010.

\bibitem{bell}
J.~S. Bell.
\newblock On the {E}instein-{P}odolsky-{R}osen paradox.
\newblock {\em Physics}, 1:195--200, 1965.

\bibitem{bennett:entEntropy}
C.~H. Bennett, H.~J. Bernstein, S.~Popescu, and B.~Schumacher.
\newblock Concentrating partial entanglement by local operations.
\newblock {\em Phys. Rev. A}, 53:2046, 1996.

\bibitem{teleport}
C.~H. Bennett, G.~Brassard, C.~Cr{\'e}peau, R.~Jozsa, A.~Peres, and
  W.~Wootters.
\newblock Teleporting an unknown quantum state via dual classical and
  {Einstein-Podolsky-Rosen} channels.
\newblock {\em Phys. Rev. Lett}, 70:1895--1899, 1993.

\bibitem{Bennett2009}
C.~H. Bennett, I.~Devetak, A.~W. Harrow, P.~W. Shor, and A.~Winter.
\newblock {Quantum Reverse Shannon Theorem}.
\newblock arXiv:0912.5537, December 2009.

\bibitem{superdense}
C.~H. Bennett and S.~Wiesner.
\newblock Communication via one- and two-particle operators on
  {Einstein-Podolsky-Rosen} states.
\newblock {\em Phys. Rev. Lett.}, 69:2881--2884, 1992.

\bibitem{mario:reverseShannon}
M.~Berta, M.~Christandl, and R.~Renner.
\newblock A conceptually simple proof of the quantum reverse shannon theorem.
\newblock arXiv:0912.3805, 2009.

\bibitem{Bhatia:97a}
R~Bhatia.
\newblock {\em {Matrix Analysis}}.
\newblock Graduate Texts in Mathematics. Springer, New York, 1997.

\bibitem{nicolas:differentResources}
N.~Brunner, N.~Gisin, and V.~Scarani.
\newblock Entanglement and non-locality are different resources.
\newblock {\em New J. of Physics}, 7:88, 2005.

\bibitem{chsh}
J.~Clauser, M.~Horne, A.~Shimony, and R.~Holt.
\newblock Proposed experiment to test local hidden-variable theories.
\newblock {\em Phys. Rev. Lett.}, 23:880--884, 1969.

\bibitem{Cleve2004}
R.~Cleve, P.~H{\o}yer, B.~Toner, and J.~Watrous.
\newblock {Consequences and limits of nonlocal strategies}.
\newblock {\em Proceedings of the 19th IEEE Conference on Computational
  Complexity}, pages 236--249, 2004.

\bibitem{CG}
D.~Collins and N.~Gisin.
\newblock A relevant two qubit {B}ell inequality inequivalent to the {CHSH}
  inequality.
\newblock {\em J. Phys. A: Math. Gen.}, 37:1775--1787, 2004.

\bibitem{Collins2002}
D.~Collins, N.~Gisin, N.~Linden, S.~Massar, and S.~Popescu.
\newblock {Bell Inequalities for Arbitrarily High-Dimensional Systems}.
\newblock {\em Phys. Rev. Lett.}, 88(4):040404, January 2002.

\bibitem{as:qmp}
A.~C. Doherty, Y.~Liang, B.~Toner, and S.~Wehner.
\newblock The quantum moment problem and bounds on entangled multi-prover
  games.
\newblock In {\em Proc. IEEE Conference on Computational Complexity}, pages
  199--210, 2008.

\bibitem{eberhard:background}
P.~Eberhard.
\newblock Background level and counter efficiencies requires for a loophole
  free {E}instein-{P}odolsky-{R}osen experiment.
\newblock {\em Phys. Rev. A (R)}, 47:747--750, 1993.

\bibitem{e91}
A.~Ekert.
\newblock Quantum cryptography based on {B}ell's theorem.
\newblock {\em Phys. Rev. Lett}, 67:661--663, 1991.

\bibitem{Froissart1981}
M.~Froissart.
\newblock {Constructive generalization of Bell’s inequalities}.
\newblock {\em Il Nuovo Cimento B}, 64(2):241--251, August 1981.

\bibitem{gisin:nonProduct}
N.~Gisin.
\newblock Bell's inequality holds for all non-product states.
\newblock {\em Phys. Lett. A}, 154:201--202, 1991.

\bibitem{marius:maxEntangled}
M.~Junge and C.~Palazuelos.
\newblock Large violation of bell inequalities with low entanglement.
\newblock arXiv:1007.3043, 2010.

\bibitem{rs:converse}
R.~K{\"o}nig and S.~Wehner.
\newblock A strong converse for classical channel coding using entangled
  inputs.
\newblock {\em Phys. Rev. Lett.}, 103:070504, 2009.

\bibitem{monique:survey}
M.~Laurent.
\newblock Sums of squares, moment matrices and optimization over polynomials.
\newblock {\em Emerging Applications of Algebraic Geometry}, 149:157--270,
  2009.

\bibitem{Liang2010}
Yeong-Cherng Liang, Tamas Vertesi, and Nicolas Brunner.
\newblock {Device-independent bounds on entanglement}.
\newblock arXiv:1012.1513, 2010.

\bibitem{yalmip}
J.~L{\"o}fberg.
\newblock Yalmip : A toolbox for modeling and optimization in {MATLAB}.
\newblock In {\em Proc. CACSD Conference}, 2004.

\bibitem{masanes05:_extrem}
Ll. Masanes.
\newblock {Extremal quantum correlations for N parties with two dichotomic
  observables per site}.
\newblock quant-ph/0512100, 2005.

\bibitem{Methot2006}
A.~A. Methot and V.~Scarani.
\newblock {An anomaly of non-locality}.
\newblock {\em QIC}, 7:157--170, January 2007.

\bibitem{npa1}
M.~Navascues, S.~Pironio, and A.~Acin.
\newblock Bounding the set of quantum correlations.
\newblock {\em Phys. Rev. Lett.}, 98:010401, 2007.

\bibitem{npa2}
M.~Navascues, S.~Pironio, and A.~Acin.
\newblock A convergent hierarchy of semidefinite programs characterizing the
  set of quantum correlations.
\newblock {\em New J. of Physics}, 10:073013, 2008.

\bibitem{Oliveira2010}
M.~Oliveira.
\newblock {Embezzlement States are Universal for Non-Local Strategies}.
\newblock arXiv:1009.0771, September 2010.

\bibitem{Pal2010}
K\'{a}roly P\'{a}l and Tam\'{a}s V\'{e}rtesi.
\newblock {Maximal violation of a bipartite three-setting, two-outcome Bell
  inequality using infinite-dimensional quantum systems}.
\newblock {\em Phys. Rev. A}, 82(2):022116, August 2010.

\bibitem{parrilo:thesis}
P.~Parrilo.
\newblock {\em Structured Semidefinite Programs and Semialgebraic Geometry
  Methods in Robustness and Optimization}.
\newblock PhD thesis, California Institute of Technology, 2000.

\bibitem{parrilo}
P.~Parrilo.
\newblock Semidefinite programming relaxations for semialgebraic problems.
\newblock {\em Math. Prog. Ser. B}, 96(2):293--320, 2003.

\bibitem{popescu:generic}
S.~Popescu and D.~R{\"o}hrlich.
\newblock Generic quantum nonlocality.
\newblock {\em Phys. Lett. A}, 166:293--297, 1992.

\bibitem{sedumi}
J.~F. Sturm.
\newblock Using sedumi 1.02, a matlab toolbox for optimization over symmetric
  cones, 1998.

\bibitem{patrick:embezzle}
W.~van Dam and P.~Hayden.
\newblock {Universal entanglement transformations without communication}.
\newblock {\em Phys. Rev. A (R)}, 67(6):060302, 2003.

\bibitem{richard:cglmp}
S.~Zohren and R.~D. Gill.
\newblock Maximal violation of the collins-gisin-linden-massar-popescu
  inequality for infinite dimensional states.
\newblock {\em Phys. Rev. Lett}, 100:120406, 2008.

\end{thebibliography}

\appendix

\section{A joint normal form for strategies using the maximally entangled state}

The goal of this section is to prove Lemma~\ref{lem:blocks}, which shows that any optimal strategy must have a certain simple joint normal form. Before we define it precisely, note that in order for the strategy $\{A_j,B_k\}_{j,k=1,\ldots,3}$ to be optimal, for a fixed choice of $\{B_k\}$ it is necessary that the operators $\{A_j\}$ be chosen so as to maximize 
\begin{align}
&\bra{\Psi} A_1 \otimes (B_1+B_2-B_3) \ket{\Psi} \label{eq:A1}\\
&\bra{\Psi} A_2 \otimes (B_1+B_2+B_3-\Id) \ket{\Psi} \label{eq:A2}\\
&\bra{\Psi} A_3 \otimes (B_2-B_1) \ket{\Psi} \label{eq:A3}
\end{align}
while for fixed $\{A_j\}$, the $\{B_k\}$ should maximize 
\begin{align}
&\bra{\Psi} B_1 \otimes (A_1+A_2-A_3-\Id) \ket{\Psi} \label{eq:B1}\\
&\bra{\Psi} B_2 \otimes (A_1+A_2+A_3-2\Id) \ket{\Psi} \label{eq:B2}\\
&\bra{\Psi} B_3 \otimes (A_2-A_1) \ket{\Psi} \label{eq:B3}
\end{align}
Since $\ket{\Psi}$ is the maximally entangled state, for any $A$ and $B$ we have $\bra{\Psi} A\otimes B \ket{\Psi} = \frac{1}{d}\Tr(AB^T) =: \langle A,B \rangle$, where $\langle \cdot,\cdot\rangle$ denotes the real Hilbert-Schmidt matrix inner product. To simplify notation, and since the $A_j$ operators always appear on the left of the tensor product (Alice's side), we will argue about $A_j^T$ rather than $A_j$, omitting the transpose sign.
Hence given for instance $B_1,B_2$ and $B_3$, the $A_1$ maximizing~\eqref{eq:A1} is simply the projector on the positive eigenspace of $B_1+B_2-B_3$.  In particular, if $B_1$, $B_2$ and $B_3$ have a joint block-diagonalization this will be reflected in $B_1+B_2-B_3$ and hence in $A_1$. This observation, combined with the CS decomposition for a pair of projectors, will let us find a simple joint form for all the $A_j$ and $B_k$, as explicited in the following definition.

\begin{definition}\label{def:normalform}
For any $c\in [-1,1]$, let $s=\sqrt{1-c^2}$ and define the $2$-dimensional projectors
\begin{align}
	P_{1}(c) &:= \frac{1}{2}\begin{pmatrix} 1-c & -s \\ -s & 1+c \end{pmatrix}\ ,\\
	P_{2}(c) &:= \frac{1}{2}\begin{pmatrix} 1-c & s \\ s & 1+c \end{pmatrix}\ ,\\
	P_3 &:= \frac{1}{2}\begin{pmatrix} 1 & 1 \\ 1 & 1 \end{pmatrix}\ .
\end{align}
We say that $d$-dimensional projectors $\{A_j,B_k\}$ are in \emph{joint normal form} if there exists a basis of $\C^d$ such that either
\begin{itemize} 
\item For even dimensions $d$, there exist reals $c_i \in [-1,1]$, $i=1,\ldots,d+1$ such that: 
\begin{itemize}
\item $A_1$ (resp. $A_2$) is block-diagonal with blocks $L_1^{2i} = P_1(c_{2i})$ (resp. $L_2^{2i} = P_2(c_{2i})$), $i=1\ldots d/2$
\item $B_3$ is block-diagonal with blocks all identical to $P_3$.
\item $B_1$ (resp. $B_2$) is block-diagonal, with the first block $R_1^1$ (resp. $R_2^1$) one-dimensional equal to $\big(\frac{1-c_{1}}{2}\big)$, the following $d/2-1$ blocks $R_1^{2i+1} = P_1(-c_{2i+1})$ (resp. $R_2^{2i+1} = P_2(-c_{2i+1})$), $i=1\ldots d/2-1$, and the last block $R_1^{d+1} = \big(\frac{1-c_{d+1}}{2}\big)$ (resp. $R_2^{d+1} = \big(\frac{1-c_{d+1}}{2}\big)$).
\item $A_3$ is block-diagonal with its first block one-dimensional equal to $(1)$, the following blocks all identical to $P_3$, and the last block one-dimensional equal to $(1)$.
\end{itemize}
\item For odd dimensions $d$, there exist reals $c_i \in [-1,1]$, $i=1,\ldots,d+1$ such that: 
\begin{itemize}
\item $A_1$ (resp. $A_2$) is block-diagonal with $(d-1)/2$ $2$-dimensional blocks $L_1^{2i} = P_1(c_{2i})$ (resp. $L_2^{2i} = P_2(c_{2i})$), $i=1\ldots (d-1)/2$, and a final $1$-dimensional block $L_1^{d+1}=\big(\frac{1-c_{d+1}}{2}\big)$ (resp. $L_2^{d+1}=\big(\frac{1-c_{d+1}}{2}\big)$), 
\item $B_3$ is block-diagonal with the first $(d-1)/2$ blocks all identical to $P_3$, and the last one $1$-dimensional equal to $(1)$.
\item $B_1$ (resp. $B_2$) is block-diagonal with an initial one-dimensional block $R_1^{1} = \big(\frac{1-c_{1}}{2}\big)$ (resp. $R_2^{1} = \big(\frac{1-c_{1}}{2}\big)$) and the following $(d-1)/2$ blocks $R_1^{2i+1} = P_1(-c_{2i+1})$ (resp. $R_2^{2i+1} = P_2(-c_{2i+1})$, $i=1\ldots (d-1)/2$.
\item $A_3$ is block-diagonal, with the first $1$-dimensional block equal to $(1)$, and all following blocks identical to $P_3$.
\end{itemize} 
\end{itemize}
Or the same as above, but with the roles of $\{A_1,A_2,B_3\}$ and $\{B_1,B_2,A_3\}$ exchanged.
\end{definition}

The main lemma of this section is the following:

\begin{lemma}\label{lem:blocks} Suppose $A_1,A_2,A_3$ and $B_1,B_2,B_3$ are six $d$-dimensional projectors achieving the maximum of~\eqref{eq:I3322} over all $d$-dimensional strategies using the maximally entangled state $\ket{\Psi}$. Then there is a $d'\leq d$, and a $d'$-dimensional strategy in joint normal form which achieves a value at least as large as that of $\{A_j,B_k\}$.
\end{lemma}

\begin{proof}
Apply the CS decomposition to $A_1$ and $A_2$, resulting in a joint block-diagonalization basis $\{\ket{e_i}\}_i$, and to $B_1$ and $B_2$, resulting in $\{\ket{f_i}\}_i$. We first show that we may take $\{\ket{e_i}\} = \{\ket{f_i}\}$ without lowering the value of the strategy.

As we already noted, the optimal choice for $A_3$ (resp. $B_3$) is the projector on the positive eigenspace of $B_2-B_1$ (resp. $A_2-A_1$). This implies that the value of~\eqref{eq:A3} does not depend on the choice of basis $\{\ket{e_i}\}$, but only on the eigenvalues of $B_2-B_1$. Hence of all the terms in~\eqref{eq:I3322}, the only ones whose value depends on the choice of the bases $\{\ket{e_i}\}$ and $\{\ket{f_i}\}$ can be grouped together as $\bra{\Psi} (A_1+A_2) \otimes (B_1+B_2) \ket{\Psi}$. 

\begin{claim}\label{claim:basis} Let $\ket{\Psi}= \frac{1}{\sqrt{d}} \sum_i  \ket{i}\ket{i}$, and $A=\sum_i \alpha_i \ket{u_i}\bra{u_i}$ and $B=\sum_i \beta_i \ket{v_i}\bra{v_i}$ positive. Then the expression $\bra{\Psi} A\otimes B \ket{\Psi}$ is maximized when the $\ket{u_i},\ket{v_i}$ are a permutation of the Schmidt basis of $\ket{\Psi}$.
\end{claim}

\begin{proof}
For any two matrices $A,B$ we have $\bra{\Psi} A \otimes B \ket{\Psi} = \frac{1}{d}\Tr(A^T B)$. 
Note that $A^T$ has the same eigenvalues as $A$. 
We then have by~\cite[Lemma IV.11]{rs:converse} that there exists a permutation $\pi \in S_d$ such that
\begin{align}
	\frac{1}{d}\tr(A^T B) \leq \sum_{j = 1}^d \lambda^A_{\pi(j)} \lambda^B_j\ ,
\end{align}
where $\lambda^{A}_1,\ldots,\lambda^{A}_d$
and $\lambda^{B}_1,\ldots,\lambda^{B}_d$ are the eigenvalues of $A$ and $B$ respectively.
\end{proof}

Given our specific choice of basis for the block-diagonalization, we have that $A_1+A_2$ (resp. $B_1+B_2$) is diagonal in the basis $\{\ket{e_i}\}$ (resp. $\{\ket{f_i}\}$), hence Claim~\ref{claim:basis} shows that these two bases may be taken equal (up to permutation) without lowering the value of the strategy. 

\medskip

We call a strategy given by projectors $\{A_j,B_k\}_{j,k}$ \emph{irreducible}
if it cannot be decomposed as a direct sum of lower-dimensional strategies. We show that any irreducible strategy has the form described in Definition~\ref{def:normalform}.

\begin{claim}\label{claim:irred} Suppose $\{A_j,B_j\}$ is irreducible. If $d$ is even, then either all blocks of the joint decomposition of $\{A_1,A_2,B_3\}$ and $\{B_1,B_2,A_3\}$ are two-dimensional, or $\{A_1,A_2,B_3\}$ have exactly two $1$-dimensional blocks and $\{B_1,B_2,A_3\}$ none (or vice-versa). If $d$ is odd, then each of $\{A_1,A_2,B_3\}$ and $\{B_1,B_2,A_3\}$ have exactly one common $1$-dimensional block.
\end{claim}

\begin{proof} We treat the case of even dimension, the odd-dimensional case being analogous. Reason by contradiction and first assume e.g. that $\{A_1,A_2,B_3\}$ each have more than two $1$-dimensional blocks in their joint block-diagonalization. We show that there is a non-trivial subspace stabilized by all operators $\{A_j,B_k\}$, contradicting the strategy's irreducibility. 

Let $\ket{e_1}$ be the vector corresponding to a one-dimensional block of $\{A_1,A_2,B_3\}$. Since the $\{\ket{f_i}\}$ are a permutation of $\{\ket{e_i}\}$, there exists an $i_1$ such that $\ket{f_{i_1}}=\ket{e_1}$. There are two possibilities for $\ket{f_{i_1}}$: either it is a joint eigenvector of $B_1,B_2$ and $A_3$ (i.e. it corresponds to a one-dimensional block in their joint block-diagonalization), or there exists an index $i_2$ such that Span$\{\ket{f_{i_1}},\ket{f_{i_2}}\}$ is left invariant by the action of $B_1,B_2$ and $A_3$ (i.e. it corresponds to a two-dimensional block). In the first case we have already found a strict subspace Span$\{\ket{e_1}\}$ stabilized by all $\{A_j,B_k\}$. In the second case we can iterate this procedure, assuming without loss of generality that $\ket{e_2}=\ket{f_{i_2}}$. There are again two cases: either $\ket{e_2}$ corresponds to a $1$-dimensional block of $\{A_1,A_2,B_3\}$, in which case Span$\{\ket{e_1},\ket{e_2}\}$ is a  non-trivial stable subspace, or there is a vector $\ket{e_3}$ such that $(\ket{e_2},\ket{e_3})$ corresponds to a $2$-dimensional block of $\{A_1,A_2,B_3\}$. We will then find an $i_3$ such that $\ket{f_{i_3}} = \ket{e_3}$, and so on.

In all cases, the process must end as soon as one of the vectors $\ket{e_k}$ encountered corresponds to a $1$-dimensional block of $\{A_1,A_2,B_3\}$. Given our assumption that there were three or more such blocks, we have found a strict subspace stabilized by all $\{A_j,B_k\}$, contradicting the irreducibility assumption.
\end{proof}

As a consequence of Claim~\ref{claim:irred}, we can block-diagonalize the pair of projectors $(A_1,A_2)$ with blocks 
\begin{align}
	L_{1}^{2i} &= \frac{1}{2}\begin{pmatrix} 1-c_{2i} & -s_{2i} \\ -s_{2i} & 1+c_{2i} \end{pmatrix}\ ,\\
	L_{2}^{2i} &= \frac{1}{2}\begin{pmatrix} 1-c_{2i} & s_{2i} \\ s_{2i} & 1+c_{2i} \end{pmatrix}\ ,
\end{align}
where $c_{2i} \in (-1,1)$ and $s_{2i} = \sqrt{1-c_{2i}^2}$, together possibly with an initial and final $1$-dimensional blocks, depending on the parity of the dimension. 

In the definition of a normal form we also require the one-dimensional blocks to have the same coefficients for both $A_1$ and $A_2$, which is
is easily seen to hold without loss of generality from the optimality of the strategy $\{A_j,B_k\}$. Indeed, let $i$ be the index of such a block, corresponding to vector $\ket{e_i}$; $A_1$ and $A_2$ are necessarily chosen so as to maximize the value of~\eqref{eq:A1} and~\eqref{eq:A2} respectively, and the coefficient in front of $(A_1)_{i,i}$ and $(A_2)_{i,i}$ will be the same in both equations, so that the optimal choice is the same. Similarly, the  matrices $(B_1,B_2)$ can be block-diagonalized with blocks:
\begin{align}\label{eq:blocks}
 R_{1}^{2i+1} &= \frac{1}{2}\begin{pmatrix} 1+c_{2i+1} & -s_{2i+1} \\ -s_{2i+1} & 1-c_{2i+1} \end{pmatrix}\ ,\\
	 R_{2}^{2i+1} &= \frac{1}{2}\begin{pmatrix} 1+c_{2i+1} & s_{2i+1} \\ s_{2i+1} & 1-c_{2i+1} \end{pmatrix}\ .
 \end{align}
Finally, it is easy to infer from~\eqref{eq:A3} (resp.~\eqref{eq:B3}) the necessary form of $A_3$ (resp. $B_3$): indeed, it is simply the projector on the positive eigenspace of $B_2-B_1$ (resp. $A_2-A_1$), which is a block $P_3$ whenever $B_1,B_2$ (resp. $A_1,A_2$) have a common $2$-dimensional block, and a block $(1)$ whenever $B_1, B_2$ (resp. $A_1,A_2$) have a common one-dimensional block. 
\end{proof}

\section{The value of a strategy in joint normal form}

In this section we derive an expression for the value obtained in~\eqref{eq:I3322} for any strategy in joint normal form (Lemma~\ref{lem:bound}), and then show how it can be upper-bounded by analytical techniques (Lemma~\ref{lem:num2}).

\begin{lemma}\label{lem:bound} Suppose $\{A_j,B_k\}$ is a strategy in joint normal form, described by a certain block structure and corresponding sequence of coefficients $c_i$. Then the value of~\eqref{eq:I3322} for this strategy for even dimensions $d$ is given by
	\begin{align}\label{eq:omegaeven}
\omega = \frac{1}{d}\sum_{i=1}^{d/2} f(c_{2i-1},c_{2i+1}) + \frac{c_1-c_{d+1}}{2d}\ ,
\end{align}
and for odd dimension $d$ by
\begin{align}\label{eq:omegaodd}
\omega &= \frac{1}{d}\sum_{i=1}^{(d-1)/2}f(c_{2i-1},c_{2i+1})\\
&\qquad + \frac{1}{d}\Big(c_d\,c_{d+1}+\frac{c_1-c_{d+1}}{2}-1+\frac{1}{2}\sqrt{1-c_d^2}\Big)\ ,\nonumber
\end{align}
where 
\begin{align}
	&f(x,y)\nonumber\\
	&= \sqrt{(x+y)^2+1}+\frac{1}{2}\sqrt{1-x^2} + \frac{1}{2}\sqrt{1-y^2} - 2\ .
\end{align}
\end{lemma}

\begin{proof} We treat the cases of even and odd dimension separately.

\paragraph{$d$ even.} In that case we know that the block-diagonalization of either $\{A_1,A_2,B_3\}$ or $\{B_1,B_2,A_3\}$ contains exactly two $1$-dimensional blocks, while the other contains none. We assume that $\{B_1,B_2,A_3\}$ has no $1$-dimensional blocks; the other case is treated symmetrically. In this case we can write
\begin{align}
	A_2 &= \frac{1}{2} \begin{pmatrix} 1-c_2 & s_2 & 0 & 0 & \cdots\\s_2 & 1+c_2 & 0 & 0 & \cdots\\ 0 & 0 & 1-c_4 & s_4 & \cdots \\ 0 & 0 & s_4 & 1+c_4 & \cdots \\ \vdots & \vdots & \vdots & \vdots & \ddots \end{pmatrix}\ ,\\
		B_2 &= \frac{1}{2} \begin{pmatrix} 1-c_{1} & 0 & 0 & 0 & \cdots\\ 0 & 1+c_3 & s_3 & 0 & \cdots\\  0 & s_3 & 1-c_3  & 0 & \cdots \\ 0 & 0 & 0&  1+c_5 & \cdots \\ \vdots & \vdots & \vdots & \vdots & \ddots \end{pmatrix}\ ,
		\end{align}
where $A_1$ and $B_1$ are identical to $A_2$ and $B_2$ respectively but have their off-diagonal elements negated, and $c_{1},c_{d+1} \in \{-1,1\}$.
 
Fixing the coefficients of $B_1$ and $B_2$, we can derive constraints on those of $A_1$ and $A_2$ from the constraint that they should be chosen so as to maximize~\eqref{eq:A1} and~\eqref{eq:A2}. The two equations are similar; let's look at~\eqref{eq:A2}. Its value can be calculated as
\begin{align}
\frac{1}{d}\sum_{i,j}& (A_2)_{i,j} \big((B_1)_{i,j}+(B_2)_{i,j}+(B_3)_{i,j} - \delta_{i,j}\big)\nonumber\\
 &= \frac{1}{d}\sum_{i=1}^{d/2}\Big( \frac{1}{2}\big(1-c_{2i}\big)\big((1-c_{2i-1})+\frac{1}{2}-1\big)\nonumber \\
 &\qquad + \frac{1}{2}\big(1+c_{2i}\big)\big((1+c_{2i+1})+\frac{1}{2}-1\big)\Big)\nonumber \\
 &\qquad + \frac{1}{d}\sum_{i=1}^{d/2} \frac{s_{2i}}{2}\nonumber\\
&= \frac{1}{d}\sum_{i=1}^{d/2}\left( \frac{1-c_{2i}}{2}\left(\frac{1}{2}-c_{2i-1}\right) \right. \nonumber\\
&\qquad \left. +\frac{1+c_{2i}}{2}\left(\frac{1}{2}+c_{2i+1}\right) + \frac{s_{2i}}{2}\right)\label{eq:tauc}
\end{align}
Setting $\tau_{2i} = (c_{2i-1}+c_{2i+1})/2$, for a fixed $c_{2i-1},c_{2i+1}$ the 
choice of $c_{2i}$ which maximizes~\eqref{eq:tauc} is $c_{2i} = 2\tau_{2i}\,(4\tau_{2i}^2+1)^{-1/2}$, which gives a value of $\frac{1}{d}\sum_{i=1}^{d/2} \sqrt{4\tau_{2i}^2+1}/2+1/2+(c_{2i+1}-c_{2i-1})/2$ for~\eqref{eq:A2}. \eqref{eq:A1} is maximized for the same choice of coefficients, and has exactly the same value. Concerning ~\eqref{eq:A3}, we find that its value is simply
\begin{align}
\frac{1}{d}\sum_{i,j} (A_3)_{i,j} \big((B_2)_{i,j}-(B_1)_{i,j}\big) &= \frac{1}{d}\sum_{i=1}^{d/2-1} s_{2i+1}\ .
\end{align}
Combining \eqref{eq:A1},\eqref{eq:A2} and~\eqref{eq:A3}, and subtracting $(1/d)(\Tr(B_1)+2\Tr(B_2))$, we obtain the value of~\eqref{eq:I3322}, which is thus
\begin{align}
	\omega &= \frac{1}{d}\sum_{i=1}^{d/2} \left(\sqrt{(c_{2i-1}+c_{2i+1})^2+1}+1\right) \nonumber \\
	&\qquad + \frac{1}{d}(c_{d+1}-c_1)+\frac{1}{d}\sum_{i=1}^{d/2-1}  \sqrt{1-c_{2i+1}^2} \nonumber \\
	&\qquad - 3\left(\frac{1}{2} + \frac{c_{d+1}-c_1}{2d}\right)\ ,
\end{align}
where we replaced $s_{2i+1} = \sqrt{1-c_{2i+1}^2}$. Using the definition of $f$, this can be re-written as
$$\omega = \frac{1}{d}\sum_{i=1}^{d/2} f(c_{2i-1},c_{2i+1})  + \frac{c_1-c_{d+1}}{2d}$$

\paragraph{$d$ odd.} In that case, each of $\{A_1,A_2,B_3\}$ and $\{B_1,B_2,A_3\}$ must have a $1$-dimensional-block in their joint block-diagonalization; say that the one for $\{A_1,A_2,B_3\}$ is the last block while the one for $\{B_1,B_2,A_3\}$ is the first block. We can proceed exactly as above to evaluate the value of this strategy, under the condition that it is optimal and hence maximizes~\eqref{eq:A1}-\eqref{eq:A3}, which lets us express the even coefficients $c_{2i}$ as a function of the odd ones $c_{2i+1}$. 
Omitting a few calculations very similar to the ones we performed in the even-dimensional case, we obtain that the value of this solution is
 \begin{align}
 \omega &= \frac{1}{d}\sum_{i=1}^{(d-1)/2} \left(\sqrt{(c_{2i-1}+c_{2i+1})^2+1}+1\right) + \frac{1}{d}\big(c_{d}-c_1\big)\nonumber\\
 &\qquad + \frac{1}{d}\big(1-c_{d+1}\big)\big(\frac{1}{2}-c_d\big)\nonumber\\
 &\qquad + \frac{1}{d}\sum_{i=1}^{(d-1)/2}  \sqrt{1-c_{2i+1}^2}  - 3\left(\frac{1}{2} - \frac{c_1}{2d}\right)\\ 
&= \frac{1}{d}\sum_{i=1}^{(d-1)/2} \left(a_i - 2\right) \\
&\qquad + \frac{1}{d}\Big(c_d\,c_{d+1}+\frac{c_1-c_{d+1}}{2}-1+\frac{1}{2}\sqrt{1-c_d^2}\Big)\ .\nonumber
\end{align}
\end{proof}

It now remains to bound $\omega$. The following claim, proven 
in Section~\ref{secApp:individualTerms}, will be useful. 

\begin{claim}\label{claim:numerics} Let $f(x,y) = \sqrt{(x+y)^2+1}+\sqrt{1-x^2}/2 + \sqrt{1-y^2}/2 - 2$ be defined on $[-1,1]^2$. Then
\begin{enumerate}
\item The maximum of $f(a,b)+f(b,c)$ over all $a,b,c\in[-1,1]^2$ such that $a+b\geq 0$ and  $b+c\leq 0$ is less than $.244$.
\item The maximum of $f(1,b)+f(b,c)$ over all $b,c\in[-1,1]^2$ such that $1+b\geq 0$ and  $b+c\leq 0$ is less than $.103$. 
\item The maximum of $f(a,1)$ over all $a\in[-1,1]$ is less than $.368$.
\end{enumerate}
\end{claim}

\begin{lemma}\label{lem:num2}
	Let $c_i \in [-1,1]$, for $i=1\ldots d+1$. Then the expression $\omega = \omega(c_i)$ in both~\eqref{eq:omegaeven} and~\eqref{eq:omegaodd} is upper-bounded by $\frac{1}{4}$.
\end{lemma}

\begin{proof}
	First note that the maximum value of the expression $c_d\,c_{d+1}+\frac{c_1-c_{d+1}}{2}-1+\frac{1}{2}\sqrt{1-c_d^2}$ over all $c_1,c_{d+1}\in\{-1,1\}$ and $c_d\in[-1,1]$ is less than $1/4$, hence~\eqref{eq:omegaodd} is always lower than~\eqref{eq:omegaeven}. Hence it is sufficient to show that $\omega = \frac{1}{d}\sum_{i=1}^{d/2} f(c_{2i-1},c_{2i+1}) + \frac{c_1-c_{d+1}}{2d}$ is upper-bounded by $1/4$, for any even $d$ and $(c_2,\ldots,c_d)\in[-1,1]^{d-1}$ and $c_1,c_{d+1}\in\{-1,1\}$. 

It is easy to verify that $f(x,y) \leq 1/2$ on the square $(x,y)\in [-1,1]^2$. Unfortunately, the extra term $\frac{c_1-c_{d+1}}{2d}$ potentially induces an additive $1/d$, so that it is not so immediate to bound $\omega$. 
Note that we can assume that $c_1 = 1$ and $c_{d+1}=-1$, since otherwise the bound follows trivially from the upper-bound on $f(x,y) \leq 1/2$. 

Given the value of $c_1$ and $c_{d+1}$, there must exist an $i$ such that $c_{2i-1}+c_{2i+1}\geq 0$ and $c_{2i+1}+c_{2i+3}\leq 0$; let $i_0$ be the first such $i$. 
We distinguish four cases, depending on the value of $i_0$.  
\begin{itemize}
\item If $d=4$, one gets that $f(1,c_3)+f(c_3,-1)<0$. Adding $(c_1-c_{d+1})/8$, one can see that $\omega < 1/4$. We assume $d>4$ for the remaining cases. 
\item If $i_0=1$, we can use the second bound in Claim~\ref{claim:numerics} to bound $f(c_1,c_3)+f(c_3,c_5)$ by $.103$, since $c_1=1$. In this case the value of $f(c_1,c_3)+f(c_3,c_5)+f(c_{d-1},c_{d+1})$ is at most $.103+.368 < .5$. Adding $1=  (c_1-c_{d+1})/2$ and dividing by $d$, we see that $\omega < 1/4$ irrespective of the value of the other $c_i$ (recall that $f(x,y)\leq 1/4$ for all $(x,y)$).  
\item If $i_0=d/2-1$, the same bound can be obtained by symmetry.
\item Otherwise $1<i_0<d/2-1$, in which case by using the first and last bounds from Claim~\ref{claim:numerics} we see that the value of $f(c_1,c_3)+f(c_{2i-1},c_{2i+1})+f(c_{2i+1},c_{2i+3})+f(c_{d-1},c_{d+1})$ is at most $.244+2\cdot .368 < 1$. Again adding $1=  (c_1-c_{d+1})/2$ and dividing by $d$, one sees that $\omega < 1/4$ irrespective of the value of the other $c_i$.
\end{itemize}
\end{proof}

\section{Details of Claim~\ref{claim:numerics}}\label{secApp:individualTerms}
We now provide the details of Claim~\ref{claim:numerics}. To find the claimed upper bounds we use a well-established 
optimization technique based on a hierarchy of semidefinite programs (SDPs) backed by the real Positivstellensatz~\cite{parrilo,parrilo:thesis}. More specifically,
if $t$ denotes a claimed upper bound, our goal will be to show that for any variables $a$, $b$ and $c$ satisfying the constraints
we have $t - h(a,b,c) \geq 0$, where $h(a,b,c)$ denotes the function we wish to optimize in case 1, 2 or 3.  To this end, we will first
rewrite any terms involving $\sqrt{\cdot}$ in the function $h(a,b,c)$ in terms of additional variables. Second, we will use
polynomial optimization techniques from~\cite{parrilo,parrilo:thesis} to obtain the bound $t$. This is exactly analogous to the techniques established in
quantum information to obtain bounds on quantum violation of Bell inequalities~\cite{npa1,npa2,as:qmp}.

We would like to emphasize that whereas semidefinite programming, as for example performed in Matlab, is a numerical technique, if a bound $t_\ell$ is obtained
at level $\ell$ of the SDP hierarchy then it is in principle possible to extract an \emph{analytical} proof that $t_\ell$ is an upper-bound on the corresponding expression $h$ from the numerics. That is, we do not rely on any
heuristic optimization methods that are not guaranteed to provide a rigorous bound.

\subsection{Case 3}

For completeness, we provide a brief informal sketch of this method for case 3; details can be found in~\cite{parrilo,parrilo:thesis}, or in the dual view of the SDP, as explained in this 
survey~\cite{monique:survey}. 
First of all, substituting 
\begin{align}
	x^2 &:= (a+1)^2 + 1 = a^2 + 2a + 2\ ,\\
	z^2 &:= 1 - a^2\ ,
\end{align}
our goal of showing that $t = 0.368$ is an upper bound to $f(a,1)$ can be restated as showing that
\begin{feasSDP}{we have}{$t \geq x + \frac{1}{2} z - 2$}
	&$x^2 = a^2 + 2a + 1$\\
	&$z^2 = 1 - a^2$\\
	&$-1 \leq a \leq 1$\ .
\end{feasSDP}
For simplicity, we will without loss of generality ignore the last constraint.
Now note that if we were able to find polynomials $t_1$ and $t_2$ in variables $x$, $z$, and $a$ such that
\begin{align}
	p &:= t - \left(x + \frac{1}{2} z - 2\right) - t_1 (a^2 + 2a + 2 - x^2)\\
	&\qquad - t_2 (1 - a^2 - z^2) = s_0\ ,\nonumber
\end{align}
where $s_0$ is a polynomial in $x$, $z$ and $a$ which is a sum of squares, then for any variables satisfying the desired constraints
$t - (x + \frac{1}{2} z - 2)\geq 0$ since $s_0$ is always positive. Our goal can thus be rephrased as searching for suitable polynomials $t_1$ and $t_2$
such that we can rewrite the resulting polynomial as a sum of squares. Very intuitively, level $\ell$ of the SDP hierarchy searches for such polynomials
up to degree $2\ell$ by searching for a matrix $Q_\ell$ such that $Q_\ell \geq 0$ and for $v_\ell = (1, a, x, z, \ldots)$ being the vector of all possible
monomials up to degree $\ell$ where we have $v_\ell^\dagger Q_\ell v_\ell = p$.
To convince ourselves, note that this means we search for $Q_\ell \geq 0$ such that
\begin{align}
	t &- \left(x + \frac{1}{2} z - 1\right) = v_\ell^\dagger Q_\ell v_\ell\\
	&\qquad + t_1 (a^2 + 2a + 2 - x^2) + t_2 (1 - a^2 - z^2)\ \nonumber
\end{align}
and thus for variables satisfying the constraints
\begin{align}
	t - \left(x + \frac{1}{2} z - 1\right) = v_\ell^\dagger Q_\ell v_\ell\ ,
\end{align}
which is clearly positive. The actual sums of squares polynomials $s_0$ can be obtained from $Q$ by diagonalizing $Q = U^\dagger D U$ where $U$ 
is unitary and $D$ is a diagonal matrix. Since $D$ only has positive entries ($Q \geq 0$), we obtain that $s_0 = \sum_j d_j (Uv)_j^\dagger (U v)_j$
is indeed a sum of squares.

It turns out that for case 3, we can already find such a matrix $Q$ at level $\ell = 0$ of the SDP, that is, $t_1, t_2 \in \Real$ are simply scalars.
To see how this works explicitly, let us first rewrite the polynomials above in terms of matrices. Let
\begin{align}
	M_0 &:= \left(\begin{array}{cccc} -2 & \frac{1}{2} & \frac{1}{4} & 0\\
		\frac{1}{2} & 0 & 0 & 0\\\frac{1}{4} & 0 & 0 & 0\\0 & 0 & 0& 0\end{array}\right)\ ,\\
	M_1 &:= \left(\begin{array}{cccc} 1 & 0 & 0 & 0\\
		0 & 0 & 0 & 0\\ 0 & 0 & -1 & 0\\0 & 0 & 0& -1\end{array}\right)\ ,
	\end{align}
	\begin{align}
	M_2 &:= \left(\begin{array}{cccc} 2 & 0 & 0 & 1\\
		0 & -1 & 0 & 0\\ 0 & 0 & 0 & 0\\1 & 0 & 0& 1\end{array}\right)\ ,\\
	T &:= \left(\begin{array}{cccc} t & 0 & 0 & 0\\
		0 & 0 & 0 & 0\\ 0 & 0 & 0 & 0\\0 & 0 & 0& 0\end{array}\right)\ .
\end{align}
Clearly, for
\begin{align}
	v &:= (1\ x\ z\ a)^T\ ,
\end{align}
we have
\begin{align}
	v^\dagger M_0 v &= x + \frac{1}{2} z - 2\ ,\\
	v^\dagger M_1 v &= 1 - a^2 - z^2\ ,\\
	v^\dagger M_2 v &= a^2 + 2a + 2 - x^2\ .
\end{align}
From the numerical solutions obtained by Matlab with SeDuMi~\cite{sedumi} and YALMIP~\cite{yalmip}, we can guess an analytical solution given by 
\begin{align}
	t_1 &= 0.51\\
	t_2 &= 0.24\\
	t &= 0.368
\end{align}
for which we can easily verify that
\begin{align}
	Q_0 := S - M_0 - t_1 M_1 - t_2 M_2 \geq 0\ ,
\end{align}
which concludes our claim.

\subsection{Cases 1 and 2}

The bounds for cases 1 and 2 are obtained analogously. The only difference is that we have to deal with more variables. 
Again, we first introduce auxiliary variables to eliminate terms containing $\sqrt{\cdot}$. We then search for suitable polynomials like $t_1$ and $t_2$ above. 
Unlike for the simple case 3, the desired bounds are not obtained at level $\ell = 0$ of the hierarchy. However, they are already found at level $\ell = 1$, and an analytical solution can again be extracted. Yet, since at level $\ell=1$ we observe polynomials of degree up to $2$ in both the original and
the auxiliary variables (in total $6$ for case 2, and $8$ for case 1) the resulting problem is already rather large (involving matrices of size $82 \times 82$ for case 1). 
We do not include these matrices here, but the Matlab scripts that can be used to extract the analytical values are available upon request.

\end{document}